\documentclass[11pt]{article}

\usepackage{amssymb} 
\usepackage{amsmath}     
\usepackage{amsthm}
\usepackage{mathtools}
\usepackage{a4wide}
\usepackage{graphicx}
\usepackage{hyperref}
\usepackage{authblk}

\usepackage{enumerate}

\newtheorem{theorem}{Theorem}
\newtheorem{lemma}{Lemma}

\newtheorem{obs}{Observation}
\newcommand{\Rset}{\mathbb{R}}

\begin{document}

\title{Recoverable robust representatives selection problem under interval continuous budgeted uncertainty}

\author[1]{Marcel Jackiewicz}
\author[1]{Adam Kasperski\footnote{Corresponding author}}
\author[1]{Pawe{\l} Zieli\'nski}

\affil[1]{
Wroc{\l}aw  University of Science and Technology, Wroc{\l}aw, Poland\\
            \texttt{\{marcel.jackiewicz,adam.kasperski,pawel.zielinski\}@pwr.edu.pl}}

\date{}

\maketitle


 \begin{abstract}
In this paper, the recoverable robust representative selection problem is considered, where uncertain second-stage costs are modeled using interval uncertainty with a continuous budget. While the variant under a discrete uncertainty budget is known to be NP-hard, we show that transitioning to a continuous budget fundamentally alters the computational complexity landscape. Specifically, by exploiting the structural properties of the problem under the continuous budget model, we design a strongly polynomial-time algorithm for the general case. Furthermore, we propose an even more efficient strongly polynomial-time algorithm for an important special case.
\end{abstract}



\maketitle


\section{Problem formulation}
  In this paper, we investigate the following two-stage version of \emph{representatives selection problem}
 under interval uncertainty.
We are given $n$ sets of items, $\mathcal{T}_1, \dots, \mathcal{T}_n$, where each 
set $\mathcal{T}_i$ contains $m_i \geq 1$ items for $i \in [n]$. Let 
$N = \sum_{i \in [n]} m_i$ denote the total number of items and 
$m = \max_{i \in [n]} m_i$, where $[n] = \{1, \dots, n\}$.
   A feasible solution contains exactly one item from each set $\mathcal{T}_i$, 
$i \in [n]$. For each item $j \in \mathcal{T}_i$, we define a first-stage cost 
$C_{ij} \geq 0$, a \emph{nominal} second-stage cost $\underline{c}_{ij} \geq 0$, and 
a maximum deviation $\Delta_{ij} \geq 0$ of the second-stage cost from its 
nominal value. In what follows, each second-stage cost $c_{ij}$ belongs to the 
interval $[\underline{c}_{ij}, \underline{c}_{ij} + \Delta_{ij}] \subseteq \mathbb{R}_{+}$. 
We use $\pmb{C} \in \mathbb{R}_{+}^N$, $\underline{\mathbf{c}} \in \mathbb{R}_{+}^N$, 
and $\pmb{\Delta} \in \mathbb{R}_{+}^N$ to denote the vectors of the 
first-stage costs, nominal second-stage costs, and deviations, respectively. 

We introduce binary variables $x_{ij}$ for $i \in [n]$ and $j \in [m_i]$, 
where $x_{ij} = 1$ if item $j \in \mathcal{T}_i$ is selected in the first stage. Let
$$\mathbb{X} = \left\{\pmb{x} \in \{0,1\}^N : \sum_{j \in [m_i]} x_{ij} = 1, \; i \in [n]\right\}$$ 
be the set of first-stage feasible solutions. Accordingly, we introduce binary 
variables $y_{ij}$ for $i \in [n]$ and $j \in [m_i]$, where $y_{ij} = 1$ if 
item $j \in \mathcal{T}_i$ is selected in the second stage. The neighborhood 
of a first-stage solution $\pmb{x} \in \mathbb{X}$ is defined as
\begin{align}
\mathbb{X}(\pmb{x}, r) = \left\{\pmb{y} \in \{0,1\}^N : \right. 
& \sum_{j \in [m_i]} y_{ij} = 1, \quad i \in [n], \nonumber\\
& \left. \sum_{i \in [n]} \sum_{j \in [m_i]} x_{ij} y_{ij} \geq r \right\}, \label{neighb}
\end{align}
where $r \in \{0, \dots, n\}$ is a given parameter. Hence, $\mathbb{X}(\pmb{x}, r)$ 
contains all second-stage solutions that share at least $r$ selected items 
with $\pmb{x}$. In the existing literature on \emph{recoverable} and \emph{incremental 
optimization} (see, e.g.,~\cite{NO13, SAO09}), three types of neighborhoods 
are commonly considered, namely \emph{element-inclusion} (at most $k$ new items can 
be added to $\pmb{x}$), \emph{element-exclusion} (at most $k$ items can be excluded 
from~$\pmb{x}$), and \emph{element-symmetric difference} (the solutions~$\pmb{x}$ 
and~$\pmb{y}$ can differ by at most $k$ items). It is easy to see that for 
the representatives selection problem, all these neighborhoods are equivalent 
to $\mathbb{X}(\pmb{x}, r)$ if we set $r = n - k$ for  the element-inclusion and element-exclusion neighborhoods;
 $r = n - 2k$  for        the element-symmetric difference neighborhood.
Therefore, in the following, 
we employ the neighborhood $\mathbb{X}(\pmb{x}, r)$.

In this paper, we use the following \emph{interval continuous budgeted uncertainty set}~\cite{NO13}, 
which contains all possible values of the second-stage costs (second-stage cost scenarios):
$$\mathcal{U}(\Gamma)=\{(\underline{\pmb{c}}+\pmb{\delta})\in \Rset_{+}^N: \pmb{0}\leq \pmb{\delta}\leq \pmb{\Delta}, \|\pmb{\delta}\|_{1}\leq \Gamma \}.$$
Hence, each second-stage cost belongs to its uncertainty interval, and the sum of 
deviations $\delta_{ij}$ from the nominal second-stage costs $\underline{c}_{ij}$ 
is bounded by a given parameter $\Gamma\in \Rset_{+}$ called a \emph{budget}.
Clearly, if $\Gamma=0$, then there is 
only one nominal cost scenario, $\underline{\pmb{c}}$, in the second stage. If $\Gamma$ 
is sufficiently large, the uncertainty set reduces to the Cartesian product of the 
uncertainty intervals, and we arrive at the classical interval uncertainty (see, e.g.,~\cite{KY97}).

In this paper, we consider a \emph{recoverable robust model}, the concept of 
which was introduced in~\cite{B11, LLMS09}. In the first stage, when the item 
costs are known, a complete solution is selected. This solution can then be 
modified to some extent after a second-stage cost scenario is revealed. We thus 
formulate the following three-level optimization problem, called
\emph{Recoverable Robust Representatives Selection}:
\begin{equation}
\text{RR-RS: }\min_{\pmb{x}\in \mathbb{X}} \left( \pmb{C}^T\pmb{x} + \max_{\pmb{c}\in \mathcal{U}(\Gamma)} \min_{\pmb{y}\in \mathbb{X}(\pmb{x},r)} \pmb{c}^T\pmb{y}\right). \label{rrrs1} 
\end{equation}
In the subsequent sections, we construct  a strongly poly\-no\-mial-ti\-me algorithm for~\eqref{rrrs1}.
Observe that~\eqref{rrrs1} contains two nested inner problems. The first is the \emph{adversarial problem}, where an adversary aims to find, for a given first-stage solution $\pmb{x} \in \mathbb{X}$, a worst-case cost scenario $\pmb{c} \in \mathcal{U}(\Gamma)$ that maximizes the recourse cost over $\mathbb{X}(\pmb{x},r)$. The second inner problem, referred to as the \emph{incremental problem}, models the decision maker's reaction: for a fixed scenario $\pmb{c}$, the goal is to find an optimal second-stage solution $\pmb{y} \in \mathbb{X}(\pmb{x},r)$ that minimizes the recovery cost, thereby effectively adjusting the first-stage decision $\pmb{x}$ to the realized cost state.

Let us briefly describe the known results on various robust representatives selection models.  
The deterministic single-stage representatives selection problem can be solved by a trivial algorithm that for each $i\in [n]$ selects an item in $\mathcal{T}_i$ with minimum cost. The problem becomes more challenging if the item costs can be uncertain. Suppose that we are given $K$ cost scenarios and we seek a solution that minimizes the maximum cost over these scenarios. This single-stage robust (min-max) representatives selection problem becomes NP-hard for $K=2$, $m_i=2$ for each $i\in [n]$, and strongly NP-hard, when $K$ is part of the input~\cite{DK12,DW13}. The latter negative result has been strengthen in~\cite{KKZ15}. In particular the problem tuned out to be hard to approximate within $O(\log^{1-\epsilon} K)$ for any $\epsilon>0$ unless NP$\subseteq$DTIME$(n^{{\rm poly}(\log  n)})$.
Fortunately, when $K$ is a constant, the problem admits a fully polynomial-time approximation scheme (FPTAS)~\cite{DK12}, and an $O(\log K / \log \log K)$-approximation algorithm when $K$ is part of the input~\cite{KKZ15}.

In the existing literature, a two-stage variant of the problem has also been investigated. Suppose that a partial solution can be constructed in the first stage when the item costs are known, and completed in the second stage after a second-stage cost scenario is revealed. This problem has been examined in~\cite{GKZ22a} for various convex uncertainty sets that model the uncertain second-stage costs.
The recoverable robust representatives selection problem under the \emph{discrete interval budgeted uncertainty}, the case where the budget~$\Gamma$ is an integer and at most $\Gamma$ second-stage costs are allowed to deviate from their nominal values,  defined in~\cite{BS04} is NP-hard. This result has been established in~\cite{GLW22}.
The representatives selection problem is a special case of the shortest path problem in series-parallel digraphs, where the sets $\mathcal{T}_i$ for $i \in [n]$ can be viewed as parallel components connected  by a series composition. The recoverable robust shortest path problem in acyclic digraphs, under interval continuous budgeted uncertainty, is known to be strongly NP-hard even in layered digraphs~\cite{KZ25}. In this paper, we identify a special case of this problem that can be solved in strongly polynomial time.

\section{Mixed integer program for RR-RS}
In this section, we construct a mixed-integer programming (MIP) formulation for~(\ref{rrrs1}). Notice first that, for a fixed scenario $\pmb{c} \in \mathcal{U}(\Gamma)$, the inner incremental (``min'') problem in~(\ref{rrrs1}) can be expressed as:
\begin{align}
	\min\quad & \sum_{i\in [n]} \sum_{j\in [m_i]} c_{ij} y_{ij} \label{mod01}\\
	\text{s.t.}\quad&\sum_{j\in [m_i]} y_{ij}=1, & i\in [n] \label{mod03c}\\
	& \sum_{i\in [n]}\sum_{j\in [m_i]} x_{ij} y_{ij}\geq r, \label{mod03}\\
	& y_{ij}\in \{0,1\}, & i\in[n], j\in [m_i]. \label{mod02}
\end{align}
Furthermore, for a fixed first-stage solution $\pmb{x} \in \mathbb{X}$, the constraint matrix of (\ref{mod03c})--(\ref{mod03}) is totally unimodular (TUM). To see this, we can apply the Ghouila-Houri characterization~\cite{C63}: a matrix $\pmb{A}$ is TUM if and only if 
for each subset of rows $\mathcal{R}$ of $\pmb{A}$ 
there exists a choice of signs $\epsilon_i\in\{1,-1\}$, $i\in \mathcal{R}$, such that for each column $j$ of $\pmb{A}$, 
$\sum_{i\in\mathcal{R}} \epsilon_i a_{ij}\in\{-1,0,1\}$.
 In our case, the rows corresponding to~(\ref{mod03c}) can be assigned the sign $+1$, while the row corresponding to~(\ref{mod03}) can be assigned the sign $-1$. 
 Accordingly,
we can relax the variables $\pmb{y}\in\{0,1\}^N$ with $\pmb{y}\in [0,1]^N$ preserving integrality of~$\pmb{y}$
and the objective function value.  

  Let $\mathbb{X}'(\pmb{x},r)$ be such a relaxed $\mathbb{X}(\pmb{x},r)$ (we replace $\pmb{y}\in\{0,1\}^N$ with $\pmb{y}\in [0,1]^N$ in~(\ref{neighb})).
 By the von Neumann minimax theorem, we can interchange the inner maximization and minimization operators in~(\ref{rrrs1}), which yields the following equivalent reformulation of the recoverable robust representatives selection problem~(\ref{rrrs1}): 
\begin{align}
	\min_{\pmb{x}\in \Phi}\; & \pmb{C}^T\pmb{x}+ \min_{\pmb{y}\in \mathbb{X}'(\pmb{x},r)}\max_{\pmb{c}\in\mathcal{U}(\Gamma)}  \pmb{c}^T\pmb{y}.\label{rrrs2} 
\end{align}
For a fixed $\pmb{y}$, the inner maximization problem is a linear program of the following form:
\begin{align*}
	\max \quad&   \sum_{i\in [n]} \sum_{j\in [m_i]} (\underline{c}_{ij}+\delta_{ij}) y_{ij}\\
		\text{s.t.} \quad& 
		\sum_{i\in [n]} \sum_{j\in [m_i]} \delta_{ij} \leq \Gamma &\\
		& 0\leq \delta_{ij} \leq \Delta_{ij}, & i\in [n], j\in [m_i].
\end{align*}
Taking the dual to this inner problem and by applying the strong duality theorem to this inner problem and substituting the result back into~(\ref{rrrs2}), we obtain the final mixed integer programming  formulation for~(\ref{rrrs1}):
\begin{small}
\begin{align}
	\min\quad & \sum_{i\in [n]}\sum_{j\in [m_i]} (C_{ij} x_{ij}+ \underline{c}_{ij}y_{ij}) \nonumber \\
	& + \sum_{i\in [n]}\sum_{j\in [m_i]} \Delta_{ij} \alpha_{ij}+\beta\Gamma \label{mipf0} \\ 
	\text{s.t.} \quad & \sum_{j\in [m_i]} x_{ij}=1, && i\in [n] \\
	& \sum_{j\in [m_i]} y_{ij}=1, && i\in [n] \label{mipf001}\\
	& \sum_{i\in [n]} \sum_{j\in [m_i]} x_{ij} y_{ij} \geq r && \label{mipf02}\\
	&  0\leq y_{ij}\leq \beta+\alpha_{ij}, &&   i\in [n], j\in [m_i] \\
	& x_{ij}\in \{0,1\}, \ \alpha_{ij}\geq 0, &&i\in [n], j\in [m_i]\label{mipf1a}\\
	&\beta \geq 0.&&\label{mipf1} 
\end{align}
\end{small}
Observe that $\beta \in [0,1]$. 

\section{Properties of the adversarial problem}
\label{secadv}
Fix a first-stage solution $\pmb{x} \in \mathbb{X}$ in~(\ref{mipf0})--(\ref{mipf1}) and consider the resulting linear program that models the adversarial problem:
\begin{small}
\begin{align}
	\min \quad&   
	\sum_{i\in [n]}\sum_{j\in [m_i]} (C_{ij} x_{ij} + \underline{c}_{ij}y_{ij}) \nonumber\\
	         & + \sum_{i\in [n]}\sum_{j\in [m_i]} \Delta_{ij} \alpha_{ij} + \beta\Gamma \label{mipa0} \\ 
		\text{s.t.} \quad& \sum_{j\in [m_i]} y_{ij} = 1, && i\in [n] \label{mipa01}\\
		& \sum_{i\in [n]} \sum_{j\in [m_i]} x_{ij}y_{ij} \geq r && \label{mipa02}\\
		& y_{ij} \leq \alpha_{ij} + \beta, && i\in [n], j\in [m_i]\label{mipa03}\\
        & y_{ij},\alpha_{ij} \geq 0, && i\in [n], j\in [m_i]\\
		& \beta \geq 0. \label{mipa10}
\end{align}
\end{small}
We now provide a characterization of the vertex solutions of the LP problem (\ref{mipa0})--(\ref{mipa10}).
 This  characterization is an application of the Rank Lemma (see, e.g.,~\cite[Lemma 1.2.3]{LRS11}). 
\begin{lemma}
\label{lrank}
Let $(\pmb{y}, \pmb{\alpha}, \beta)$ be any vertex solution to LP~\eqref{mipa0}--\eqref{mipa10} such that $\beta>0$ and let
\begin{align*}
\mathcal{Y} &= \{y_{ij} > 0 \mid i \in [n], \, j \in [m_i]\}, \\
\mathcal{L} &= \{\alpha_{ij} > 0 \mid i \in [n], \, j \in [m_i]\}.
\end{align*}
Denote by  $W$ the set of all tight constraints at this vertex solution chosen among the cardinality constraint~\eqref{mipa02} and the bound constraints~\eqref{mipa03}.
Then there exists a subset of tight constraints $W' \subseteq W$ such that:
\begin{enumerate}[(i)]
    \item $|\mathcal{Y}| + |\mathcal{L}| + 1 = n + |W'|.$\label{lranki}
    \item The system $\mathcal{E}$ of $n + |W'|$ linear equations corresponding to the $n$ equality constraints~\eqref{mipa01} and the selected tight constraints in $W'$, restricted to the strictly positive variables of $(\pmb{y}, \pmb{\alpha}, \beta)$, is linearly independent. \label{lrankii}
    \item The positive components of $(\pmb{y}, \pmb{\alpha}, \beta)$ are the unique solution to the tight linear system $\mathcal{E}$. \label{lrankiii}
\end{enumerate}
\end{lemma}
\begin{proof}
By the Rank Lemma, the positive components of the vertex solution $(\pmb{y}, \pmb{\alpha}, \beta)$ are uniquely characterized by any maximal set of linearly independent constraints that are tight with respect to this solution, where the cardinality of this set is exactly $|\mathcal{Y}| + |\mathcal{L}| + 1$. 
Clearly, the inequality $n + |W| \geq |\mathcal{Y}| + |\mathcal{L}| + 1$ must hold, as the number of tight constraints cannot be smaller than the number of positive variables. We construct the set $\mathcal{E}$ greedily, exploiting the properties of the vector matroid  (see, e.g.,~\cite[Section 13]{KV12}).

Initially, we include all $n$ linearly independent equality constraints from~\eqref{mipa01} into $\mathcal{E}$. We then iteratively add constraints from $W$ to $\mathcal{E}$, while preserving the linear independence of the system, until a maximal number of linearly independent tight constraints from $W$ is obtained. Note that if the cardinality constraint~\eqref{mipa02} is tight but linearly dependent on the already chosen constraints, it remains in $W \setminus W'$. Since $\mathcal{E}$ consists of the $n$ constraints from~\eqref{mipa01} and  $|W'|$ constraints from $W'$, we have $|\mathcal{E}| = n + |W'|$. Finally, by the Rank Lemma, the number of strictly positive variables must equal the number of linearly independent tight equations, which yields $|\mathcal{Y}| + |\mathcal{L}| + 1 = n + |W'|$, thereby completing the proof.
\end{proof}

The following lemma provides a structural characterization  of the vertex solutions of the LP problem (\ref{mipa0})--(\ref{mipa10})
and 
is crucial for constructing an algorithm for the problem~(\ref{rrrs1}):
\begin{lemma}
\label{lemvar}
Let $(\pmb{y}, \pmb{\alpha}, \beta)$ be any vertex solution of the LP problem (\ref{mipa0})--(\ref{mipa10}). Then the possible values of the variables $y_{ij}$ and $\beta$ are characterized as follows:
\begin{align}
    &y_{ij} \in \{0, \beta, 1-\ell\beta\},  \quad \ell \in \{0, \dots, m-1\}, \label{casey1}\\
    &\beta \in \left\{1, \frac{1}{2}, \dots, \frac{1}{m},0\right\} \quad \text{or} \quad \beta = \frac{r-l_2}{l_1-l_3},  \label{beta}
\end{align}
where $l_1, l_2 \in \{0, \dots, n\}$, $l_1 + l_2 \leq n$, $l_3 \in \{0, \dots,  l_2(m-1)\}$, $m=\max_{i\in[n]}m_i$
\end{lemma}
\begin{proof}
Choose any vertex solution $(\pmb{y}, \pmb{\alpha}, \beta)$ to the linear program~\eqref{mipa0}--\eqref{mipa10}. 
Assume first that $1 \ge \beta > 0$. 
The positive components of $(\pmb{y}, \pmb{\alpha}, \beta)$ constitute the unique solution to the tight, square, and nonsingular linear system~$\mathcal{E}$ of size $(|\mathcal{Y}| + |\mathcal{L}| + 1) \times (n + |W'|)$, satisfying $|\mathcal{Y}| + |\mathcal{L}| + 1 = n + |W'|$, by Lemma~\ref{lrank}(\ref{lranki}),~(\ref{lrankii}),~(\ref{lrankiii}), 
where $W' \subseteq W$ is the subset of the tight constraints in $\mathcal{E}$, beyond the $n$ equality constraints~\eqref{mipa01}.
Below, we show how to explicitly determine the values of $(\pmb{y}, \beta)$ in this solution by analyzing the algebraic structure of the system~$\mathcal{E}$.
 
%
Each variable $\alpha_{ij} \in \mathcal{L}$ uniquely specifies a tight constraint of the form $y_{ij} = \alpha_{ij} + \beta$ in $W'$ (observe that $\alpha_{ij} > 0 \iff y_{ij} > \beta$). Let $W'_{\pmb{\alpha}} \subseteq W'$ denote the set of these tight constraints, so $|W'_{\pmb{\alpha}}| = |\mathcal{L}|$. The set of remaining tight bound constraints of the form $y_{ij} = \beta$ (where $\alpha_{ij} = 0$) in $W'$ is denoted by $W'_{\pmb{y}}$. In the next part of the proof, we  use the following observation: 
\begin{obs}
\label{obsbeta}
There is at most one equality constraint among~(\ref{mipa01}) in which all variables belonging to $\mathcal{Y}$ are explicitly fixed to $\beta$ by the bound constraints in $W'_{\pmb{y}}$. 
\end{obs}
Indeed, if there are two such equality constraints, then they form a linearly dependent or infeasible  system with the constraints in $W'_{\pmb{y}}$, which contradicts the definition of $\mathcal{E}$  (see Lemma~\ref{lrank}(\ref{lrankii}),~(\ref{lrankiii})). 

We consider two cases depending on the membership of the cardinality constraint~\eqref{mipa02} in $W'$.

\emph{Case 1: Constraint \eqref{mipa02} is not in $W'$.}
This implies that constraint~\eqref{mipa02} is either not tight or linearly dependent on other constraints in $\mathcal{E}$. In this case, $W'$ consists exclusively of tight constraints from~\eqref{mipa03}, so $W' = W'_{\pmb{\alpha}} \cup W'_{\pmb{y}}$ and $|W'| = |\mathcal{L}| + |W'_{\pmb{y}}|$. Substituting this into Lemma~\ref{lrank}(\ref{lranki}) yields:
\[
|\mathcal{Y}| + |\mathcal{L}| + 1 = n + |\mathcal{L}| + |W'_{\pmb{y}}| \implies |\mathcal{Y}| - |W'_{\pmb{y}}| = n - 1.
\]
This relation establishes that among the $|\mathcal{Y}|$ variables in $\mathcal{Y}$, exactly $|W'_{\pmb{y}}|$ variables are explicitly fixed to $\beta$ by the bound constraints in $W'_{\pmb{y}}$. Consequently, exactly $n - 1$ variables in $\mathcal{Y}$ are unfixed by the bound constraints in $W'_{\pmb{y}}$. We  now analyze the values of these $n-1$ variables and the variable $\beta$ across the $n$  equality constraints~\eqref{mipa01}:
\begin{enumerate}[(a)]
    \item\label{c1a} \emph{The value of $\beta$:} 
    Distributing $n - 1$ unfixed variables across $n$  equality constraints~\eqref{mipa01} implies 
    the existence of \emph{exactly one} constraint~$i'$ (see Observation~\ref{obsbeta}) in which all variables belonging to $\mathcal{Y}$ are explicitly fixed to $\beta$ by constraints in $W'_{\pmb{y}}$ (i.e., $y_{i'j} = \beta$ for all $j$ with $y_{i'j} \in \mathcal{Y}$). 
    Summing over the constraint~$i'$ in~\eqref{mipa01} yields:
    $\sum_{j \in [m_{i'}]} y_{i'j} = \ell_{i'} \beta = 1 \implies \beta = \frac{1}{\ell_{i'}}$,
    where $\ell_{i'} \in \{1, \dots, m_{i'}\}$ is the number of variables in the constraint~$i'$ fixed by the bound constraints in $W'_{\pmb{y}}$.
    Since $m_{i'} \le m$, this proves that $\beta \in \left\{1, \frac{1}{2}, \dots, \frac{1}{m}\right\}$, establishing~\eqref{beta}.

    \item \emph{The values of the remaining $n-1$ variables $y_{ij}$:} 
    Since constraint~$i'$ is the unique one with zero unfixed variables (see part~(\ref{c1a})), the remaining $n - 1$ unfixed variables are distributed across the $n - 1$ remaining constraints~(\ref{mipa01}). 
    Every such an equality constraint $i \neq i'$ contains \emph{exactly one} variable in $\mathcal{Y}$ unfixed by the bound constraints in $W'_{\pmb{y}}$. All other $\ell_i$ variables in constraint~$i$, belonging to $\mathcal{Y}$, are fixed by the bound constraints in $W'_{\pmb{y}}$ and equal $\beta$. Thus:
 $y_{ij} + \ell_i \beta = 1 \implies y_{ij} = 1 - \ell_i \beta, \quad \text{where } \ell_i \in \{0, \dots, m_i - 1\}$.
\end{enumerate}
Accordingly, every variable in $\mathcal{Y}$ satisfies $y_{ij} \in \{\beta, 1 - \ell_i \beta\}$ for some $\ell_i \in \{0, \dots, m_i - 1\}$,
$m_i\leq m$. Together with $y_{ij} = 0$ for $y_{ij} \notin \mathcal{Y}$, this completes the proof of~\eqref{casey1}.

\emph{Case 2: Constraint \eqref{mipa02} is in $W'$.}
In this case, constraint~\eqref{mipa02} is tight and belongs to $W'$, so $W' = \{\eqref{mipa02}\} \cup W'_{\pmb{\alpha}} \cup W'_{\pmb{y}}$ and $|W'| = 1 + |\mathcal{L}| + |W'_{\pmb{y}}|$. Substituting this into Lemma~\ref{lrank}(\ref{lranki}) gives:
\[
|\mathcal{Y}| + |\mathcal{L}| + 1 = n + 1 + |\mathcal{L}| + |W'_{\pmb{y}}| \implies |\mathcal{Y}| - |W'_{\pmb{y}}| = n.
\]
This implies that exactly $n$ variables in $\mathcal{Y}$ are not explicitly fixed to $\beta$ by the bound constraints in $W'_{\pmb{y}}$. 
We analyze the values of these $n$ variables and the variable $\beta$ across the $n$  equality constraints~\eqref{mipa01}
and the cardinality constraint~\eqref{mipa02}. We need to consider two subcases.

\emph{Case 2a:} Each  equality constraint~\eqref{mipa01}  contains \emph{exactly one} variable in $\mathcal{Y}$ unfixed by $W'_{\pmb{y}}$.
 Consequently, the structural characterization of $y_{ij}$ remains identical to Case~1: every variable in $\mathcal{Y}$ satisfies $y_{ij} \in \{\beta, 1 - \ell_i \beta\}$ for some $\ell_i \in \{0, \dots, m_i - 1\}$, $m_i\leq m$.

Constraint~\eqref{mipa02} is used to uniquely determine $\beta$. Substituting the parameterized forms of $y_{ij}$ into the tight constraint~\eqref{mipa02} (with fixed binary choices $x_{ij} \in \{0,1\}$) yields:
\begin{equation}
l_1\beta + \underbrace{(1-\upsilon_1 \beta) + (1-\upsilon_2 \beta) + \dots + (1-\upsilon_{l_2} \beta)}_{l_2 \;\text{times}} = r,\label{econ}
\end{equation}
where $l_1, l_2 \in \{0, \dots, n\}$, $l_1 + l_2 \leq n$, and $\upsilon_l \in \{0, \dots, m_l-1\}$ for $l \in [l_2]$. Equivalently, this can be written as $l_1\beta + l_2 - (\upsilon_1 + \upsilon_2 + \dots + \upsilon_{l_2})\beta = r$, which simplifies to:
$$(l_1 - l_3)\beta = r - l_2,$$
where $l_3 = \sum_{l=1}^{l_2} \upsilon_l \in \{0, \dots, l_2(m-1)\}$  and $l_1, l_2, l_3$ represent the contribution of the remaining 
 constraints to the cardinality constraint.
 If $l_1 = l_3$, then the system $\mathcal{E}$ is either infeasible (if $r-l_2 \neq 0$) or linearly dependent (if $r-l_2 = 0$)
 --  see Lemma~\ref{lrank}(\ref{lrankii}),~(\ref{lrankiii}). Hence, we obtain:
$$\beta = \frac{r - l_2}{l_1 - l_3},$$
and thus \eqref{beta} is proven.

\emph{Case 2b:}  Not every equality constraint~\eqref{mipa01} contains exactly one variable in $\mathcal{Y}$ unfixed by $W'_{\pmb{y}}$.
By Observation~\ref{obsbeta}, there must exist \emph{exactly one} constraint $i_1 \in [n]$ with two unfixed variables, say $y_{i_1 j_1}, y_{i_1 j_2} \in \mathcal{Y}$, and \emph{exactly one} constraint $i_2 \in [n]$ with zero unfixed variables.
The characterization of the remaining $n-2$ unfixed variables   is  the same as in Case~1 due to the fact that each of these variables
appears in exactly one equality constraint). Now they become fixed.
Accordingly, 
the relevant system of tight constraints reduces to:
\begin{align}
    & y_{i_1 j_1} + y_{i_1 j_2} = 1 - \ell_{i_1}\beta, \label{rsys1}\\
    & \ell_{i_2}\beta = 1, \label{rsys2} \\
    & x_{i_1 j_1} y_{i_1 j_1} + x_{i_1 j_2} y_{i_1 j_2} = r - l_2 - (l_1 - l_3)\beta, \label{rsys3}
\end{align}
where $\ell_{i_2} \in \{1, \dots, m_{i_2}\}$, $\ell_{i_1} \in \{0, \dots, m_{i_1}-2\}$,
 $l_3=\sum_{l\in [l_2]\,:l\not=  j_{1}, j_{2}}\upsilon_l$, $\ell_{i_1}$ and $\ell_{i_2}$ stand for
 the numbers of fixed variables by $W'_{\pmb{y}}$ in constraints~$i_1$ and $i_2$, respectively. 
 Here $l_1, l_2, l_3$ represent the contribution of the rest
 constraints to the cardinality constraint (see~\eqref{econ}).
 By the linear independence of the system~$\mathcal{E}$, the columns corresponding to $y_{i_1 j_1}$ and $y_{i_1 j_2}$ must be linearly independent, which implies $x_{i_1 j_1} \neq x_{i_1 j_2}$. Without loss of generality, assume that $x_{i_1 j_1} = 1$ and $x_{i_1 j_2} = 0$. This yields:
\begin{align}
   & y_{i_1 j_2} = 1 - \ell_{i_1}\beta - y_{i_1 j_1}, \label{yij2}\\
   &  \beta = \frac{1}{\ell_{i_2}}, \label{betazg} \\
    & y_{i_1 j_1} = r - l_2 - (l_1 - l_3)\beta. \label{yij1}
\end{align}
Since $1 = \ell_{i_2}\beta$, substituting this into \eqref{yij1} yields:
\begin{equation}
y_{i_1 j_1} = (r - l_2)\ell_{i_2}\beta - (l_1 - l_3)\beta = M\beta,\label{ey11}
\end{equation}
where $M = (r - l_2)\ell_{i_2} - (l_1 - l_3)$ is a \emph{positive integer}.
Substituting \eqref{ey11} into \eqref{yij2} and using the equation~(\ref{betazg}) in~(\ref{ey11}) gives:
\begin{eqnarray}
y_{i_1 j_2} = 1 - (\ell_{i_1} + M)\beta.\label{ey12} \\
y_{i_1 j_1} = 1 - (\ell_{i_2} - M)\beta\label{ey13}
\end{eqnarray}
where $\ell_{i_1}+M\in \mathbb{Z}$ and $\ell_{i_2}-M\in \mathbb{Z}$.
Since $y_{i_1 j_2} \in (0,1)$, $1 \le \ell_{i_1} + M \le \ell_{i_2} - 1$, where the second inequality results from $1-(\ell_{i_1}+M)\beta=1-(\ell_{i_1}+M)\frac{1}{\ell_{i_2}}>0$ and $\ell_{i_2}>0$.
Accordingly, because $y_{i_1 j_1} \in (0,1)$, and $M$ is a positive integer,  we get $1 \le \ell_{i_2} - M \le \ell_{i_2} - 1$. As $\ell_{i_2}\leq m_{i_2}\leq m$, both $y_{i_1 j_1}$ and $y_{i_1 j_2}$ are of the form~(\ref{casey1}).
This completes the proof of \eqref{casey1} and \eqref{beta} for $1\geq \beta > 0$.

It remains to consider the case where $\beta = 0$. Observe that in any vertex solution, we then have $\alpha_{ij} = y_{ij}$, and thus the variables $\alpha_{ij}$, together with the constraints \eqref{mipa03}, can be eliminated. The constraint matrix of the remaining model (with fixed $x_{ij} \in \{0,1\}$) is totally unimodular (TUM). Therefore, in any vertex solution, $y_{ij} \in \{0,1\}$, which is entirely encompassed by the characterization in \eqref{casey1}.
\end{proof}

\section{Polynomial-time  algorithm for RR-RS}

We are now ready to present a strongly polynomial-time algorithm for solving the model \eqref{mipf0}--\eqref{mipf1} and, consequently, the problem \eqref{rrrs1}. The core idea of the algorithm is to solve the inner model \eqref{mipf0}--\eqref{mipf1a} for a fixed $\beta \in [0,1]$ and retrieve an optimal solution to model \eqref{mipf0}--\eqref{mipf1} corresponding to a value of $\beta$ that minimizes the objective~\eqref{mipf0}. By Lemma~\ref{lemvar}, we can restrict our attention to the discrete set of candidate values given in \eqref{beta}. Let $\mathcal{B}$ denote this set of candidate values for $\beta$.

Fix $\beta \in \mathcal{B}$. Observe that, in the absence of the cardinality constraint \eqref{mipf02}, the inner model \eqref{mipf0}--\eqref{mipf1a} for a fixed $\beta$ decomposes into $n$ independent subproblems, one for each set $\mathcal{T}_i$, $i \in [n]$. Constraint \eqref{mipf02} couples these subproblems. Thus, to solve \eqref{mipf0}--\eqref{mipf1a}, we need to enumerate all possible cases for constraint \eqref{mipf02} and combine them by solving the shortest path problem on a suitably constructed layered digraph. In the following, we provide the complete details.

For a given set $\mathcal{T}_i$, $i \in [n]$, and item $p \in \mathcal{T}_i$, Lemma~\ref{lemvar} implies that it suffices to consider the assignment $\gamma \in \{0, \beta, 1-\ell\beta\}$, where $\ell \in \{0, \dots, m-1\}$. Under these conditions, we formulate the following subproblem:
 \begin{align}
d_{ip}^\gamma=\min\quad&   
	\sum_{j\in [m_i]}  \underline{c}_{ij}y_{ij}\nonumber \\
	 &+ \sum_{j\in [m_i]} \Delta_{ij} \alpha_{ij}  \label{mipb0} \\ 
		\text{s.t.} \quad& \sum_{j\in [m_i]} y_{ij} = 1 \label{mipb01}\\
		& y_{ij} \leq \alpha_{ij} + \beta, &&  j\in [m_i]\label{mipb2}\\
        & y_{ij},\alpha_{ij} \geq 0, && j\in [m_i] \\
        & y_{ip}=\gamma.  \label{mipb4}     
\end{align}
The LP model \eqref{mipb0}--\eqref{mipb4} optimally allocates values to the variables $y_{ij}$ in the set $\mathcal{T}_i$, provided that $y_{ip}=\gamma$, where $d_{ip}^\gamma$ is the cost of this optimal allocation. 
Notice that in any optimal solution to \eqref{mipb0}--\eqref{mipb4}, since $\Delta_{ij} \geq 0$ and the objective function is minimized, the variable $\alpha_{ij}$ will take the smallest possible value satisfying \eqref{mipb2}. Thus, $\alpha_{ij} = \max\{0, y_{ij} - \beta\}$. To linearize this, we apply the variable splitting technique by substituting $y_{ij} = \underline{y}_{ij} + \overline{y}_{ij}$ for each $j \in [m_i]$. Here, $\underline{y}_{ij} \in [0, \beta]$ represents the baseline portion of $y_{ij}$ associated with the nominal cost $\underline{c}_{ij}$, while $\overline{y}_{ij} \geq 0$ explicitly represents the excess $\alpha_{ij}$ that incurs the additional penalty $\Delta_{ij}$. Furthermore, since $\sum_{j\in[m_i]} y_{ij} = 1$, we obtain the upper bound $\overline{y}_{ij} \leq 1 - \beta$. 
This yields the following equivalent reformulation for the subproblem \eqref{mipb0}--\eqref{mipb4}:
\begin{align}
d_{ip}^\gamma=\min \quad &   
	\sum_{j\in [m_i]} \underline{c}_{ij}\underline{y}_{ij} + \sum_{j\in [m_i]} (\underline{c}_{ij}+\Delta_{ij}) \overline{y}_{ij}  \label{mipc0} \\ 
		\text{s.t.} \quad& \sum_{j\in [m_i]\setminus\{p\}} (\underline{y}_{ij}+\overline{y}_{ij}) = 1-\gamma \label{mipc1}\\
		& 0\leq \underline{y}_{ij} \leq \beta, \quad  j\in [m_i] \label{mipc2}\\
		& 0\leq \overline{y}_{ij} \leq 1-\beta, \quad j\in [m_i]. \label{mipc3}
\end{align}
Note that after fixing the variable $ y_{ip}= \underline{y}_{ip}+\overline{y}_{ip}=\gamma$, the model \eqref{mipc0}--\eqref{mipc3} is a continuous knapsack problem that can be solved in $O(m_i)$ (see, e.g., \cite[Corollary 17.5]{KV12}). 
Let us define:
\begin{equation}
f_i^\gamma = \min_{p\in [m_i]} \{C_{ip} + d_{ip}^\gamma\} \text{ and }
j_i^\gamma \in \arg\min_{p\in [m_i]} \{C_{ip} + d_{ip}^\gamma\},\label{fcost}
\end{equation}
where $\gamma \in \{0, \beta, 1-\ell\beta\}$ and $\ell \in \{0, \dots, m-1\}$. The value $f_i^\gamma$ represents the minimum cost of allocating the value $\gamma$ to one of the items in $\mathcal{T}_i$, and $j_i^\gamma \in [m_i]$ is the index of the item for which this minimum is attained. 
Note that computing this minimum corresponds exactly to determining the optimal selection of the binary variable $x_{i j_i^\gamma} = 1$ within the set $\mathcal{T}_i$, under the assumption that its associated continuous variable takes the value $y_{i j_i^\gamma} = \gamma$.
To evaluate the values $f_i^\gamma$ for all $p \in [m_i]$ and 
 $\gamma \in \{0, \beta, 1-\ell\beta\}$, $\ell \in \{0, \dots, m-1\}$, 
one can avoid solving the continuous knapsack instances naively in $O(m^3)$. By exploiting the monotonicity of the residual capacity $1 - \gamma$, the computation for 
set~$\mathcal{T}_i$ can be efficiently carried out using a two-pointer sliding window technique combined with precomputed prefix sums. This approach reduces the overall time complexity for $\mathcal{T}_i$ to~$O(m^2)$.

We now construct a layered digraph $G^{\beta}=(V,A)$. The node set $V$ is partitioned into $n+2$ layers, denoted as $V_0, \ldots, V_{n+1}$. Each node in layer $V_{i}$, for $i\in [n]$, represents a state defined by the accumulated value of the left-hand side of the cardinality constraint \eqref{mipf02} after $i$ decisions.

Layer $V_0$ contains only the source node $s$, initialized with a label of $0$. By Lemma~\ref{lemvar}, 
the  assignment~$\gamma_i$ takes a value from $\{0, \beta, 1-\ell\beta\}$, where $\ell \in \{0, \dots, m-1\}$. Consequently, any accumulated label $L_i = \sum_{v=1}^i \gamma_v$ at layer $V_{i}$ can be expressed in the form:
$L_i = l_1 + (l_2 - l_3)\beta$,
where $l_1$ is the number of times a value of the form $1-\ell_v\beta$, $v\in [i]$, was chosen, $l_2$ is the number of times $\beta$ was chosen, and $l_3 = \sum_{v\in [i]} \ell_v$, where the sum is taken over the $l_1$ selections. 
Clearly, $l_1 \in \{0, \dots, i\}$ and $l_2 \in \{0, \dots, i\}$. The number of possible values for $l_1$ is $O(i)$, and the number of possible values for $l_2 - l_3$ is $O(i \cdot m)$. 
Thus, the number of distinct node labels in any layer $V_{i}$ is $O(i^2m)$ and each layer can contain up to $O(n^2m)$ nodes.

For each node in layer $V_i$, where $0 \leq i \leq n-1$, with label~$L_{i}$ and for each assignment 
$\gamma \in \{0, \beta, 1-\ell\beta\}$,  $\ell \in \{0, \dots, m-1\}$, we create an arc to a node in layer $V_{i+1}$ with label $L_{i} + \gamma$. The cost of this arc is set to $f_{i+1}^\gamma$. This forward construction is applied iteratively up to layer $V_{n}$. Finally, the terminal layer $V_{n+1}$ consists only of the sink node~$t$ with a label of $r$. We introduce a zero-cost arc from a node in $V_{n}$ with label $L_{n}$ to~$t$ if and only if $L_n \geq r$.
A digraph construction is illustrated in Figure~\ref{fig2}.

\begin{figure*}[ht]
\centering
\includegraphics{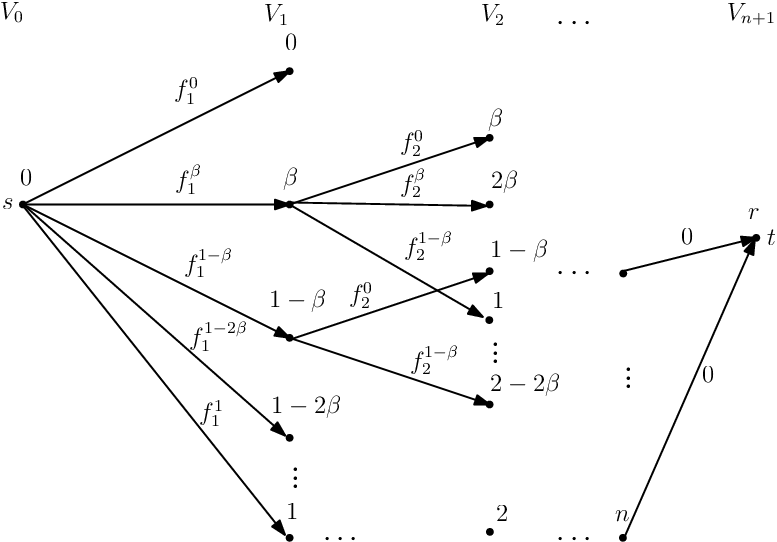}
\caption{Illustration of the layered digraph $G^\beta=(V,A)$ construction (for clarity, only some arcs in $A$ are shown).
Node label represent the accumulated values of the left-hand side of constraint~\eqref{mipf02} denoted by $L_i$. Arcs from $V_{i-1}$ to $V_{i}$ encode the local allocations at cost $f_i^\gamma$, while valid terminal arcs to the sink~$t$ are formed exclusively from nodes in $V_{n}$
with labels~$L_n$ satisfying $L_n\geq r$.} \label{fig2}
\end{figure*}

The size of $G^\beta = (V, A)$ is bounded as follows: 
the number of nodes in layer $V_i$ is $|V_{i}| = O(n^2 m)$, the total number of nodes is $|V| = O(n^3 m)$, 
the out-degree of any node in layer $V_i$ is $O(m)$,
and the total number of arcs is $|A| = O(n^3 m^2)$.

\begin{lemma}
\label{lsp}
There exists a feasible solution to \eqref{mipf0}--\eqref{mipf1} for a fixed $\beta$ with a total cost of at most $c^\beta + \beta \Gamma$ if and only if there exists an $s$-$t$ path $P$ in $G^\beta$ with a total cost of at most~$c^\beta$.
\end{lemma}
\begin{proof}
$(\Rightarrow)$ Assume there exists a feasible solution $(\pmb{x}, \pmb{y}, \pmb{\alpha})$ to \eqref{mipf0}--\eqref{mipf1} with total cost at most $c^\beta + \beta \Gamma$. For each set $\mathcal{T}_i$, $i \in [n]$, exactly one binary variable is active, say $x_{i j'_i} = 1$ for some $j'_i \in [m_i]$, and its corresponding continuous variable takes the value $y_{i j'_i} = \hat{\gamma}_i$. By Lemma~\ref{lemvar}, $\hat{\gamma} \in \{0, \beta, 1-\ell\beta\}$ for some $\ell_i \in \{0, \dots, m-1\}$. By the construction of $G^\beta$, there exists an arc in layer $i$ corresponding to allocation $\hat{\gamma}_i$ with cost $f_i^{\hat{\gamma}_i} \le C_{i j'_i} + d_{i j'_i}^{\hat{\gamma}_i}$. Furthermore, since constraint \eqref{mipf02} is satisfied, we have $\sum_{i=1}^n \hat{\gamma}_i \geq r$, which guarantees the existence of a zero-cost arc from the node in layer $V_{n}$ to the sink $t$. Concatenating these arcs yields an $s$-$t$ path $P$ in $G^\beta$ whose total cost satisfies $\sum_{i=1}^n f_i^{\hat{\gamma}_i} \leq c^\beta$.

$(\Leftarrow)$ Let $P$ be an $s$-$t$ path in $G^\beta$ with a total cost of~$c^\beta$. Let $j_1^{\gamma_1}, \dots, j_n^{\gamma_n}$ be the item indices associated with the arcs along this path (see~\eqref{fcost}), where $\gamma_i \in \{0, \beta, 1-\ell\beta\}$ with $\ell \in \{0, \dots, m-1\}$ for each $i \in [n]$. We construct a feasible solution to \eqref{mipf0}--\eqref{mipf1} by fixing $x_{i j_i^{\gamma_i}} = 1$ and $y_{i j_i^{\gamma_i}} = \gamma_i$ for all $i \in [n]$. The remaining continuous variables in each set $\mathcal{T}_i$ are optimally allocated by solving the subproblem \eqref{mipb0}--\eqref{mipb4}. Since $P$ is a valid $s$-$t$ path terminating at node $t$, the node labels satisfy $L_n=\sum_{i=1}^n \gamma_i \geq r$, ensuring that the cardinality constraint \eqref{mipf02} holds. By the definition of the arc costs $f_i^{\gamma_i}$, the objective function value of this constructed solution is precisely $\sum_{i=1}^n f_i^{\gamma_i} + \beta \Gamma = c^\beta + \beta \Gamma$.
\end{proof}

We are now ready to summarize the complete algorithm for solving the recoverable robust representative selection problem \eqref{mipf0}--\eqref{mipf1}. For each $\beta \in \mathcal{B}$, we construct the corresponding layered digraph $G^\beta$ and compute a shortest $s$-$t$ path in $G^\beta$. Let $c^\beta$ denote the total cost of this shortest path. 
By Lemma~\ref{lsp}, $c^\beta + \beta\Gamma$ is the optimal value of the objective function of the inner model \eqref{mipf0}--\eqref{mipf1a}.
Finally, we select an optimal parameter value:
$$\beta^* \in \operatorname*{argmin}_{\beta \in \mathcal{B}} \{c^\beta + \beta\Gamma\}$$
and retrieve an optimal solution to the inner model \eqref{mipf0}--\eqref{mipf1a} corresponding to $\beta^*$, which is an optimal solution to \eqref{mipf0}--\eqref{mipf1} and thus to the RR-RS problem.

We now analyze the running time of the algorithm. The digraph $G^\beta$ has $O(n^3 m)$ nodes and $O(n^3 m^2)$ arcs. We can compute the costs $f_i^{\gamma}$ and indices $j_{i}^{\gamma}$ for all possible assignments $\gamma$ (see \eqref{fcost}) in $O(m^2)$ time for each $i \in [n]$. Hence, all costs $f_{i}^{\gamma}$ and the corresponding indices $j_{i}^{\gamma}$ in $G^\beta$ can be found in $O(n m^2)$ time. 
It follows that $G^\beta$ together with all arc costs can be constructed in $O(n^3 m^2 + n m^2)$ time. Since $G^\beta$ is an acyclic digraph, computing a shortest path takes $O(|V| + |A|) = O(n^3 m^2)$ time (see, e.g.,~\cite[Section 7]{KV12}). 
Finally, we have to construct~$G^\beta$, compute arc cost and  
solve the shortest path problem on~$G^\beta$  for each candidate value for $\beta\in \mathcal{B}$. Thus, the overall running time of the algorithm is $O(|\mathcal{B}|(n^3 m^2 + n m^2))$, where $|\mathcal{B}| = O(n^2 m)$, which gives a total time complexity of $O(n^5 m^3)$.
We have thus proved the following theorem.
\begin{theorem}
There exists a strongly polynomial-time algorithm for the RR-RS problem with
the element-inclusion, element-exclusion and element-symmetric difference neighborhoods. 
\end{theorem}

\section{The RR-RS problem without bounds on the second-stage costs}
In this section, we investigate a special case of the RR-RS problem where $\Delta_{ij} = \infty$ for all $i \in [n]$ and $j \in [m_i]$. Consequently, no upper bound is imposed on the second-stage cost increases. We show that this assumption enables us to decrease  the running time of an algorithm for this problem. First, note that model~\eqref{mipf0}--\eqref{mipf1} then reduces to the following formulation:
\begin{small}
\begin{align}
	\min& \sum_{i\in [n]}\sum_{j\in [m_i]} (C_{ij} x_{ij}+ \underline{c}_{ij}y_{ij})+\beta\Gamma&& \label{mipfu0} \\ 
	\text{s.t. } & \sum_{j\in [m_i]} x_{ij}=1, &&  i\in [n] \\
	& \sum_{j\in [m_i]} y_{ij}=1, &&   i\in [n] \label{mipfu001}\\ 
	& \sum_{i\in [n]} \sum_{j\in [m_i]} x_{ij} y_{ij} \geq r && \label{mipfu02}\\
	&  0\leq y_{ij}\leq \beta, &&    i\in [n], j\in [m_i] \\
	& x_{ij}\in \{0,1\},  &&  i\in [n], j\in [m_i]\label{mipfu1a}\\
	&\beta \geq 0.&&\label{mipfu1} 
\end{align}
\end{small}
and the corresponding adversarial problem, for a fixed first-stage solution~$\pmb{x}$, becomes
\begin{small}
\begin{align}
	\min\quad & \sum_{i\in [n]}\sum_{j\in [m_i]} \underline{c}_{ij}y_{ij}+\beta\Gamma \label{mipau0} \\ 
	\text{s.t.} \quad 
	& \sum_{j\in [m_i]} y_{ij}=1, &&   i\in [n] \label{mipau001}\\ 
	& \sum_{i\in [n]} \sum_{j\in [m_i]} x_{ij} y_{ij} \geq r && \label{mipau02}\\
	&  y_{ij}\leq \beta, &&   i\in [n], j\in [m_i] \label{mipau1a}\\
	&  y_{ij}\geq 0, &&   i\in [n], j\in [m_i] \label{mipau2a}\\
	&\beta \geq 0.&&\label{mipau1} 
\end{align}
\end{small}
We now provide a structural characterization of the vertex solutions of the LP problem (\ref{mipau0})--(\ref{mipau1}). 
\begin{lemma}
\label{lemvaru}
Let $(\pmb{y},\beta)$ be any vertex solution of the LP problem (\ref{mipau0})--(\ref{mipau1}). Then the possible values of the variables $y_{ij}$ and $\beta$ are characterized as follows:
\begin{align}
    &y_{ij} \in \{0, \beta, 1-\ell^\beta\beta \},  \quad \ell^\beta= \left\lfloor \frac{1}{\beta} \right\rfloor, \label{caseyu1}\\
    &\beta \in \left\{1, \frac{1}{2}, \dots, \frac{1}{\underline{m}},0\right\} \quad \text{or} \quad \beta = \frac{r-l_2}{l_1-l_3},  \label{betau}
\end{align}
where $l_1, l_2 \in \{0, \dots, n\}$, $l_1 + l_2 \leq n$, $l_3\in \{0, \dots,  l_2\underline{m}\}$, $\underline{m}=\min_{i\in[n]}m_i$.
\end{lemma}
\begin{proof}
Model~\eqref{mipau0}--\eqref{mipau1} is a special case of model~\eqref{mipa0}--\eqref{mipa10} with $y_{ij} \le \beta$
($\alpha_{ij}=0$, $i\in[n]$, $j\in[m_i]$). Therefore, the proof is similar to the one of Lemma~\ref{lemvar}.
Let $(\pmb{y}, \beta)$ be an arbitrary vertex solution to LP~\eqref{mipau0}--\eqref{mipau1}.
Assume that $1 \ge \beta > 0$.
By Lemma~\ref{lrank},
the positive components of $(\pmb{y},  \beta)$ are the unique solution to the tight, square, and nonsingular
 linear system~$\mathcal{E}$ of size $(|\mathcal{Y}| + 1) \times (n + |W'|)$, satisfying $|\mathcal{Y}| + 1 = n + |W'|$,
 here $\mathcal{L}=\emptyset$.
 We recall that $W'_{\pmb{y}}\subseteq W'$ is the set of tight bound constraints of the form $y_{ij} = \beta$.
 Below, we show how to explicitly determine the values of $(\pmb{y}, \beta)$. 
 It is clear that Observation~\ref{obsbeta} remains true for the considered case.
 Observe also that $\beta\geq \frac{1}{\underline{m}}$ due to the  equality constraints~\eqref{mipau001} and
 the fact that $y_{ij}\leq \beta$. We need to consider only two cases.

\emph{Case 1: Constraint \eqref{mipau02} is not in $W'$.}
In this case, $W'=W'_{\pmb{y}}$. Thus $|\mathcal{Y}|-|W'_{\pmb{y}}|=n-1$,
which means that exactly $n - 1$ variables in $\mathcal{Y}$ are unfixed by bound constraints in $W'_{\pmb{y}}$.
Furthermore, there is 
 \emph{exactly one} equality constraint~$i'$ in which all variables are fixed by the  bound constraints $W'_{\pmb{y}}$ and
 each of $n-1$ remaining equality constraints $i \neq i'$ contains \emph{exactly one} variable in $\mathcal{Y}$ unfixed by the bound constraints in $W'_{\pmb{y}}$ (see Case 1 in the proof of  Lemma~\ref{lemvar}).
 
Substituting the variables fixed  by the constraints in $W'_{\pmb{y}}$  into equality constraint~$i'$
yields:
$\ell_{i'} \beta = 1 \implies \beta = \frac{1}{\ell_{i'} }$,
where $\ell_{i'}  \in \{1, \dots, \underline{m}\}$ is the number of fixed variables in constraint~$i'$. This proves~\eqref{betau}.

We now determine values of the remaining $n-1$ variables $y_{ij}$.
Consider $i \neq i'$ containing an unfixed  variable~$y_{ij}$. 
Constraint~$i$ then yields $y_{ij} + \ell_i \beta = 1 \implies y_{ij} = 1 - \ell_i \beta$.
Since $0 < y_{ij} \leq\beta$, substituting $y_{ij} = 1 - \ell_i \beta$ into this bounds inequality gives $0 < 1 - \ell_i \beta \leq \beta \iff \ell_i < \frac{1}{\beta} \leq \ell_i + 1$, which determines $\ell_i = \left\lfloor \frac{1}{\beta} \right\rfloor = \ell^\beta$. Thus, $y_{ij} = 1 - \ell^\beta \beta$.
 Accordingly, every variable in $\mathcal{Y}$ satisfies $y_{ij} \in \{\beta, 1 - \ell^\beta \beta\}$. Together with $y_{ij} = 0$ for $y_{ij} \notin \mathcal{Y}$, this completes the proof of~\eqref{caseyu1}.

\emph{Case 2: Constraint \eqref{mipau02} is in $W'$.}
In this scenario, the cardinality constraint~\eqref{mipau02} is tight and belongs to the linearly independent system $\mathcal{E}$,
 $W' = W'_{\pmb{y}} \cup \{\eqref{mipau02}\}$. Applying Lemma~\ref{lrank}(\ref{lranki}) gives $|\mathcal{Y}| - |W'_{\pmb{y}}| = n$.
This indicates that exactly $n$ variables from $\mathcal{Y}$ are unfixed by  the bound constraints in $W'_{\pmb{y}}$. 
We investigate  the values of these $n$ variables and the variable $\beta$ across the $n$  equality constraints~\eqref{mipau001}
and the cardinality constraint~\eqref{mipau02}.  Two subcases must be considered.

\emph{Case 2a:} Each  equality constraint  contains \emph{exactly one} variable in $\mathcal{Y}$ unfixed by $W'_{\pmb{y}}$.
Therefore,
 the structural characterization of $y_{ij}$ is  the same as in Case~1: every variable in $\mathcal{Y}$ satisfies
 $y_{ij} \in \{\beta, 1 - \ell^\beta \beta\}$.
 Substituting these values into the tight constraint~\eqref{mipau02} 
 with fixed $x_{ij} \in \{0,1\}$ gives
\[
l_1 \beta + l_2 (1 - \ell^\beta \beta) = r \implies (l_1 - l_3)\beta = r - l_2,
\]
where $l_1, l_2 \in \{0, \dots, n\}$ with $l_1 + l_2 \le n$ and $l_3 = l_2 \ell^\beta$.
Noting that $\ell^\beta \leq \underline{m}$ implies $l_3 \in \{0, \dots, l_2 \underline{m}\}$.
If $l_1 = l_3$, then the system $\mathcal{E}$ is either infeasible (if $r-l_2 \neq 0$) or linearly dependent (if $r-l_2 = 0$).
This contradicts Lemma~\ref{lrank}(\ref{lrankii}),~(\ref{lrankiii}). Thus, we have:
$$\beta = \frac{r - l_2}{l_1 - l_3},$$
which proves formula \eqref{betau}.

\emph{Case 2b:}  Not every equality constraint  contains exactly one variable in $\mathcal{Y}$ unfixed by $W'_{\pmb{y}}$.
In this case the reasoning is the same as in Case~2b in the proof of  Lemma~\ref{lemvar}.
Namely, 
there exists \emph{exactly one} constraint $i_2 \in [n]$ with zero unfixed variables and \emph{exactly one} constraint $i_1 \in [n]$ with two unfixed variables, say $y_{i_1 j_1}, y_{i_1 j_2} \in \mathcal{Y}$. The remaining $n-2$ constraints each contain exactly one unfixed variable, characterized as in Case~1. Now they are fixed.
Consequently, the tight linear system reduces to the one analogous to~\eqref{rsys1}--\eqref{rsys3}. Without loss of generality, assume $x_{i_1 j_1} = 1$ and $x_{i_1 j_2} = 0$.Because $y_{i_1 j_1} = M\beta$ (see ~\eqref{ey11}), $y_{i_1j_1}\leq \beta$ and $M$ is a positive integer, we get $M=1$. Therefore, $y_{i_1j_1}=\beta$ and $y_{i_1j_2}=1-\ell^\beta\beta$.
This completes the proof of~\eqref{caseyu1} and~\eqref{betau} for $1 \ge \beta > 0$.

The proof of the case, where $\beta=0$ is the same as in the proof of  Lemma~\ref{lemvar}.
\end{proof}

The algorithmic framework for solving the RR-RS problem without bounds on the second-stage costs mimics that of the bounded case. 
The core approach is to solve the inner model (\ref{mipfu0})-(\ref{mipfu1a})  for a fixed $\beta \in [0,1]$ and retrieve an optimal solution to model (\ref{mipfu0})-(\ref{mipfu1a}) corresponding to a value of $\beta$ that minimizes the objective~(\ref{mipfu0}). 
 By Lemma~\ref{lemvaru}, it suffices to consider the discrete set of candidate values specified in \eqref{betau}. 
For each candidate value of $\beta$ from Lemma~\ref{lemvaru}, we construct a layered graph $G^\beta$ and solve a shortest path problem. However, the size of $G^\beta$ is now substantially reduced. Specifically, each layer $V_i$ contains up to $O(n^2)$ nodes, resulting in a total of $O(n^3)$ nodes in $G^\beta$. Since at most three arcs leave each node, $G^\beta$ also contains $O(n^3)$ arcs. Given that the number of candidate values for $\beta$ is $O(n^2\underline{m})$, this leads
to an algorithm with the running time 
  of $O(n^5\underline{m})$.





\end{document}